\documentclass[12pt]{article}

\usepackage{amsmath,amssymb,amsthm}
\usepackage[margin=1in]{geometry}
\usepackage[numbers,super,sort&compress]{natbib}
\usepackage{subcaption}
\usepackage{graphicx}
\usepackage{booktabs}   
\usepackage{chngcntr}
\usepackage{setspace}
\usepackage{hyperref}
\usepackage{authblk}
\usepackage[export]{adjustbox}
\usepackage[percent]{overpic}
\usepackage{comment}

\newtheorem{prop}{Proposition}

\newtheorem{definition}{Definition}

\title{Bounding Causal Effects for Ordinal Outcomes Under Positive Dependence}

\author{Micha Mandel\thanks{micha.mandel@mail.huji.ac.il} }
\author{Daniel Rodan}

\affil{Department of Statistics and Data Science, The Hebrew University of Jerusalem,
Mount Scopus, Jerusalem 9190501, Israel}

\begin{document}

\date{\textbf{Short title:} Causal Effects for Ordinal Outcomes}

\maketitle
\thispagestyle{empty}

\newpage
\thispagestyle{empty}

\begin{abstract}

\setstretch{1.2}

Defining and estimating causal effects for ordinal data is challenging. Standard average treatment effects are not appropriate for ordinal scales, and alternative estimands, such as the probabilities that the treatment outcome exceeds or does not worsen the control outcome, are generally not identifiable. Existing work provides sharp bounds for these quantities based only on marginal distributions. Motivated by a previous observation showing that bounds obtained under an independence working assumption can be substantially tighter, we investigate conditions under which such bounds are valid. We show that commonly used notions of positive dependence, including positive quadrant dependence and positive regression dependence, are not sufficient to justify these bounds. We then propose a new dependence condition, diagonal tail dominance (DTD), under which the independence-based bounds are guaranteed to hold. We explain why this condition is quite strong and may not be appropriate in many settings, limiting the justification for using the independence-based bounds. However, local DTD may be plausible in many applications, and we derive improved bounds that exploit an independence working assumption on selected parts of the probability table. Through theoretical results, numerical examples, and an analysis of data from a clinical trial of a new treatment for acute ischemic stroke, we illustrate the properties of the bounds and the role of the proposed conditions.

\vspace{1cm}
\end{abstract}

\clearpage
\setcounter{page}{1}

\setstretch{1.2}

\section{Introduction}

Ordinal outcomes arise frequently in educational, behavioral, and public health research, for example in responses to Likert-type questionnaire items (e.g., strongly disagree to strongly agree) or in clinical classifications such as cancer stage I--IV. In such settings, the categories possess a clear ordering, yet the numerical distances between adjacent categories are not meaningfully defined. This distinguishes ordinal variables from nominal outcomes, for which ordering is irrelevant, as well as from interval or ratio measurements, for which differences and averages carry substantive meaning \citep{agresti2010analysis}.

The ordered but non-metric nature of ordinal outcomes poses challenges for both description and inference. Although ordinal variables encode more information than nominal categories, treating them as quantitative responses by assigning numerical scores may lead to conclusions that depend on arbitrary scaling choices and lack clear interpretation \citep{rebchuk2020health}. Consequently, standard causal parameters based on averages and differences are generally inappropriate for ordinal data, and alternative estimands are required.

We consider the two-sample setting in which an ordinal outcome is compared between a treatment group ($1$) and a control group ($0$). Using the potential outcomes framework \citep{rubin1974estimating}, let $Y(0)$ and $Y(1)$ denote the potential outcomes for a randomly selected individual in the study population: $Y(0)$ is the outcome if the individual receives the control, and $Y(1)$ is the outcome if she receives the treatment. We assume that the pair $(Y(0),Y(1))$ has a joint distribution $P$ with marginal distributions $P_0$ and $P_1$.

The standard average treatment effect is $E\{Y(1)-Y(0)\}$, but, as noted above, this estimand is not well defined for ordinal outcomes. Several alternative measures have been proposed for ordinal outcomes, including probabilities comparing the treatment and control potential outcomes and measures based on stochastic dominance. These measures compare the marginal distributions $P_1$ and $P_0$ or probabilities derived from the joint distribution $P$, such as $P\{Y(1)>Y(0)\}$. Because the joint distribution of $(Y(0),Y(1))$ is not observable, estimands that depend on $P$ are generally not identifiable, and bounds must instead be derived. Such bounds have been derived previously \citep{huang2017inequality,chiba2017sharp,lu2018treatment,cheng2009estimation,lu2020sharp}, but the interval between the lower and upper bounds, within which the true estimand lies, is often very wide. Lu et al.\cite{lu2018treatment}~suggest using independence-based lower bounds for  $P\{Y(1)>Y(0)\}$ and  $P\{Y(1)\ge Y(0)\}$ under the assumption of positively correlated potential outcomes. While these bounds can substantially improve upon the model-free bounds, the precise positivity assumption required for their validity is not discussed. 

We start by clarifying geometrically the structure of the problem, and visually shows which entries of $P$ are captured by the bounds and which are excluded. We then revisit the independence-based bounds proposed for positively dependent potential outcomes \citep{lu2018treatment} and show that even strong notions of dependence are not sufficient to guarantee the validity of these bounds. We identify a new property of the joint distribution $P$ under which these bounds are valid, and a weaker property under which we develop new tighter lower bounds. Thus, the paper contributes both theoretical insight and practical tools for causal inference with ordinal outcomes.

We focus on bounds based on marginal probabilities and on theoretical properties of $P$, without considering important issues such as confounding, compliance, and inference that have been extensively studied \citep{cheng2009estimation}. We assume that the marginal distributions are identifiable and can be well estimated from the data.

The paper is organized as follows. Section \ref{sec:estimands} discusses different estimands for causal inference with ordinal data that have been proposed in the literature. Section \ref{sec:bounds} presents model-free bounds and illustrates their validity using colored matrices. Section \ref{sec:DTD} presents the tighter bounds of Lu et al.\cite{lu2018treatment}~for positively dependent potential outcomes. It shows, through numerical examples, that standard dependence notions do not guarantee the validity of these bounds and introduces a new alternative sufficient dependence condition. It also derives improved lower bounds under a more plausible dependence assumption. Section \ref{sec:real_data} illustrates the bounds using real data from a clinical trial of a new treatment for acute ischemic stroke. Section \ref{sec:discussion} concludes the paper with a discussion of extensions to continuous data and comments on inference.

\section{Model-free causal estimands for ordinal outcomes} \label{sec:estimands}

Assume that the ordinal outcome takes values in $\{1,\ldots,J\}$, where a higher value indicates a better outcome. A common approach to define a causal estimand is to assign scores or utilities $s_1<\cdots<s_J$ to the ordinal levels and to calculate the standard average treatment effect as
\begin{equation} \label{eq:sATE}
E\{S(1)\}-E\{S(0)\},
\end{equation}
where $(S(0),S(1))$ are the potential scores corresponding to $(Y(0),Y(1)),$ and $E$ is the expectation operator taken over the relevant population \citep{benkeser2021improving}. Assigning numerical scores to the categories amounts to an attempt to transform the ordinal outcome into an interval-scale variable, by imposing meaningful distances between adjacent outcome levels. Examples of studies using scores are Strich et al.\cite{strich2022fostamatinib}~who study a new treatment for hospitalized Covid-19 patients and use the mean difference in a simple ranking of an ordinal severity outcome, and Ratcliff et al.\cite{ratcliff2023early}~who study a new invasive surgery approach for intracerebral hemorrhage using as estimand the mean difference of a utility function applied to the ordinal modified Rankin score.

Replacing the ordinal levels with scores is often arbitrary and can complicate interpretation. The resulting estimand depends on the chosen scoring system, and the resulting treatment effect may appear substantial under one set of scores but small, or even absent, under another. This concern is illustrated by Rebchuk et al.\cite{rebchuk2020health}~who found substantial variability in utility weights for the modified Rankin Scale.

An opposite approach is binarization that collapses the ordinal outcome into two groups defined by a selected threshold. A vector estimand is the distributional causal effect vector \citep{ju2010criteria}, whose $j$'th component is
\[
\Delta_j = P\{Y(1)\ge j\}-P\{Y(0)\ge j\},
\]
${j=1,\ldots,J}$. It measures the difference in the cumulative probabilities of achieving level $j$ or higher under treatment and control. Figure \ref{fig:delta} shows the joint distribution of $(Y(0),Y(1))$, with the events corresponding to $P\{Y(1)\ge j\}$ (left) and $P\{Y(0)\ge j\}$ (middle) highlighted in green. The estimand $\Delta_j$ is the difference between these two quantities and is shown in the right panel for $j=3$, where green cells indicate probabilities to be summed and orange cells indicate probabilities to be subtracted. As seen in the figure, it focuses on part of the support points and therefore loses much of the information in an ordinal outcome, especially if it has many levels.

\begin{figure}[tbp]
\centering
\begin{minipage}{0.333\textwidth}
\centering
\includegraphics[width=\linewidth]{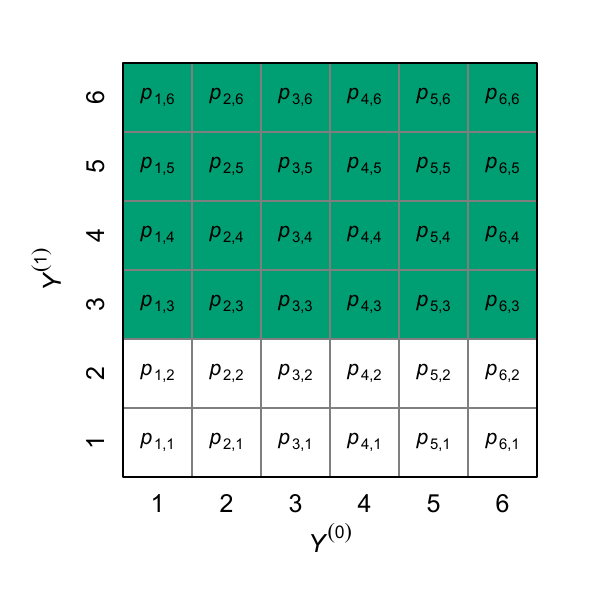}
\end{minipage}%
\begin{minipage}{0.333\textwidth}
\centering
\includegraphics[width=\linewidth]{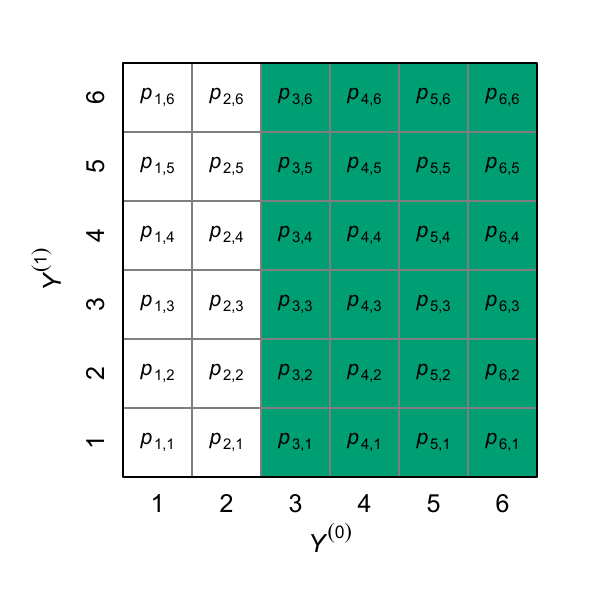}
\end{minipage}%
\begin{minipage}{0.333\textwidth}
\centering
\includegraphics[width=\linewidth]{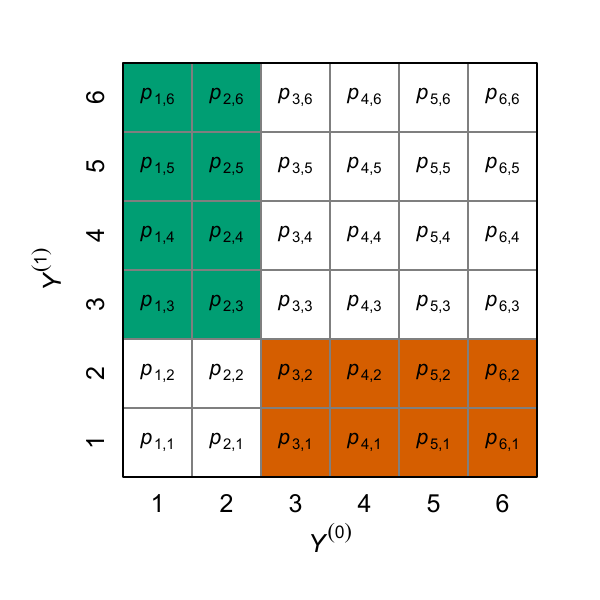}
\end{minipage}
\caption{\fbox{\parbox{0.95\textwidth}{%
Estimands for an ordinal outcome with six levels. The matrix represents the joint probability table of $(Y(0),Y(1))$. The colored cells in the left and middle panels correspond to the events $P\{Y(1)\ge 3\}$ and $P\{Y(0)\ge 3\}$, respectively. The right panel shows $\Delta_3$, defined as the difference between these two probabilities, where green cells indicate probabilities to be summed and orange cells indicate probabilities to be subtracted.
}}}
%
\label{fig:delta}
\end{figure}

There are $J-1$ possible contrasts $(\Delta_1\equiv0$ by definition), and a weighted average of the form
\begin{equation} \label{eq:wDelta}
\sum_{j=1}^J w_j[P\{Y(1)\ge j\}-P\{Y(0)\ge j\}],
\end{equation}
where $w_1,\ldots,w_J$ are positive weights, can be used to summarize the causal effect. While this measure depends only on the marginal distributions and is therefore identifiable, it depends on the arbitrary selection of the set of weights  \citep{chiba2017sharp}. In fact, it is easy to show, using the tail-sum representation of the expectation, that the weighted average in \eqref{eq:wDelta} is equivalent to \eqref{eq:sATE} when $s_0=0$ and $w_j=s_j-s_{j-1}$, or equivalently $s_j=\sum_{\ell\le j} w_\ell$ ($j=1,\ldots,J$). 

Alternatively, one may consider causal measures defined in terms of the joint distribution $P$ of $(Y(0),Y(1))$. 
Lu et al.\cite{lu2018treatment}~study bounds for the causal estimands:
\begin{equation} \label{eq:etau}
    \tau = P\{Y(1)\ge Y(0)\}\quad \rm{and}\quad  \eta = P\{Y(1)>Y(0)\},
\end{equation}
which measure the fraction of individuals who are not worse off and who benefit from treatment, respectively. Panel (a) of Figure \ref{fig:bounds} depicts in green the cells corresponding to $\tau$ on the joint distribution table of $(Y(0),Y(1))$; the entries corresponding to $\eta$ are obtained by removing the diagonal of the table. The estimands are the sums of probabilities of these cells. Because the joint distribution of the potential outcomes is not observable, these estimands are generally not identifiable and must instead be bounded.

\begin{figure}[tbp]
\centering
\begin{overpic}[width=0.43\textwidth]{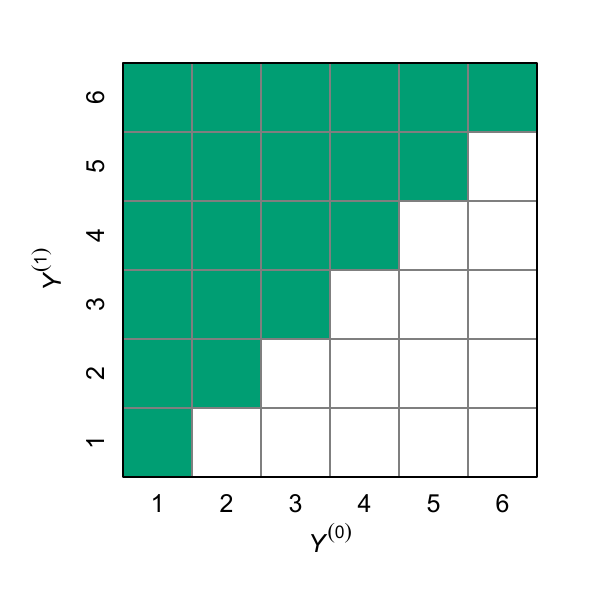}
  \put(2,92){\textbf{(a)}}
\end{overpic}
\hfill

\vspace{-0.3em}

\begin{overpic}[width=0.43\textwidth]{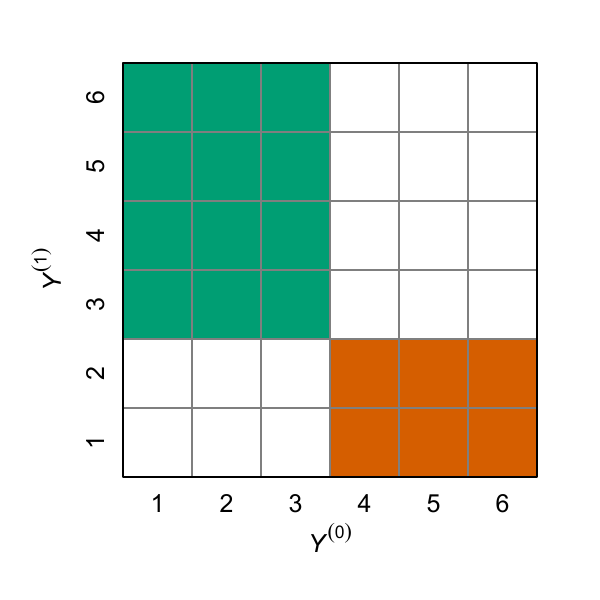}
  \put(2,92){\textbf{(b1)}}
\end{overpic}
\hfill
\begin{overpic}[width=0.43\textwidth]{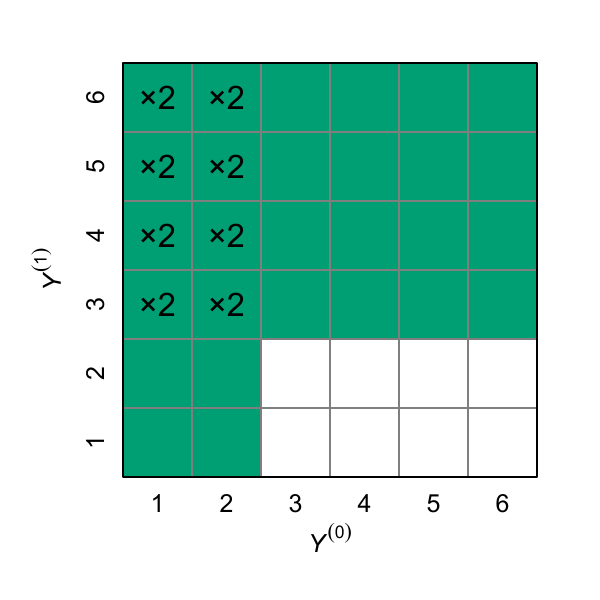}
  \put(2,92){\textbf{(b2)}}
\end{overpic}

\vspace{-0.3em}

\begin{overpic}[width=0.43\textwidth]{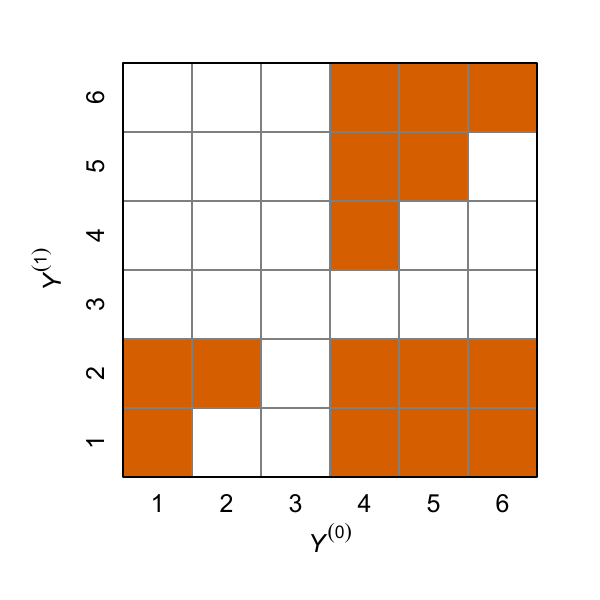}
  \put(2,92){\textbf{(c1)}}
\end{overpic}
\hfill
\begin{overpic}[width=0.43\textwidth]{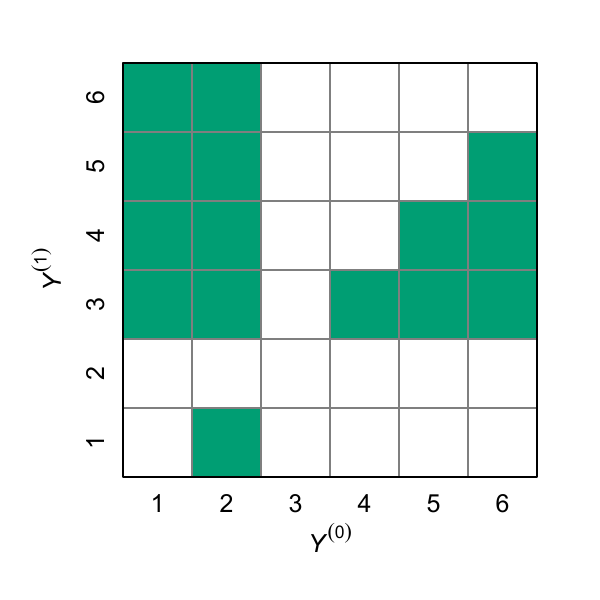}
  \put(2,92){\textbf{(c2)}}
\end{overpic}
\caption{\setlength{\fboxsep}{6pt}\fbox{\parbox{0.95\textwidth}{%
Probability of functionals on the joint probability table of $(Y(0),Y(1))$. Entries corresponding to $\tau$ are shown in \textbf{(a)}. Candidate lower \textbf{(b1)} and upper \textbf{(b2)} bounds for $\tau$ are shown for $j=3$. The bounds sum the green entries and subtract the orange ones. Entries marked with $\times 2$ are summed twice. The differences between the lower \textbf{(c1)} and upper \textbf{(c2)} bounds and $\tau$ are shown.
}}}
\label{fig:bounds}
\end{figure}


Huang et al.\cite{huang2017inequality}~develop bounds for the fraction of individuals who benefit from treatment, $\eta$, under a more general setting in which $P$ may have structure, for example when certain cells of $P$ are structural zeros. Other related estimands include stochastic superiority of the treatment outcome $\eta + 0.5P\{Y(1)=Y(0)\}$ \citep{cheng2009estimation}, the relative treatment effect $P\{Y(1)>Y(0)\}-P\{Y(0)>Y(1)\} = \tau+\eta-1$ \citep{lu2020sharp}, and a modified $\tau$ estimand that subtracts the worst outcome $\tau-P\{Y(1)=Y(0)=0\}$ \citep{chiba2017sharp}.

While each measure has its own advantages and limitations, they are generally not identifiable from the marginal distributions alone and require either additional assumptions for identification or replacement by bounds. In this paper we focus on the estimands $\eta$ and $\tau$ given in \eqref{eq:etau}, and revisit the positive-dependence assumptions on $P$  used in order to improve the bounds\cite{lu2018treatment}.

\section{Model-free bounds for $\tau$ and $\eta$} \label{sec:bounds}

Lu et al.\cite{lu2018treatment}~provide sharp model-free bounds for $\tau$ and $\eta$ through a direct algebraic analysis of the probability matrix. We illustrate the validity of their bounds using colored matrices, which provide a convenient graphical tool for this purpose. We emphasize that the bounds must be identifiable, that is, being functionals of the marginal distributions alone.

\begin{prop}  \label{prop:tau}
    (Lu et al.\cite{lu2018treatment}~Proposition 1.) The sharp lower and upper bounds of $\tau$ are
    \begin{equation}
        \tau_L = \max_{1\le j\le J}\{P\{Y(0)=j\} + \Delta_j\} \qquad
        \tau_U = \min_{1\le j \le J} \{1+\Delta_j\}.
    \end{equation}
\end{prop}

We refer to the terms inside the minimum and maximum operators as candidate bounds, since each provides a valid bound for the estimand, and the sharp bounds are obtained by taking the minimum or maximum over these quantities. Panels (b1) and (b2) of Figure \ref{fig:bounds} illustrate the lower and upper candidate bounds of $\tau$ corresponding to $j=3$ for an outcome with six levels; see Figure \ref{fig:delta} for $\Delta_3$ and Panel (a) of Figure \ref{fig:bounds} for $\tau$ for the same model. Web Figures 1 and 2 depict all the lower and upper candidate bounds for $\tau$ for that model, respectively. Specifically, Panel (b1) depicts $P\{Y(0)=3\}+\Delta_3$.  The lower candidate bound for $\tau$ is obtained by summing the green entries in the colored matrix and subtracting the orange entries. As is readily seen, this quantity is smaller than $\tau$ (see Panel (a)). Since a similar construction applies for each $j$, taking the maximum over $j$ yields a lower bound. Panel (b2) depicts $1+\Delta_3$. The bound is obtained by summing the green entries, with those marked by $\times 2$ counted twice. This quantity clearly bounds $\tau$ from above. Since an analogous construction applies for each $j$, taking the minimum over $j$ yields the upper bound. 

The matrices illustrate how the bounds differ from $\tau$ by showing which entries are captured and which are not. Panels (c1) and (c2) of Figure \ref{fig:bounds} depict the differences between the bounds and $\tau$ (obtained by subtracting Panel (a) from Panels (b1) and (b2)). Only when the probabilities in these colored cells are small the bounds are close to $\tau$. It is important to emphasize that Figure \ref{fig:bounds} presents both the lower and upper candidate bounds for $j=3$ while the global upper and lower bounds are  obtained for different cutoffs $j$; see Web Figures 1 and 2 for all lower and upper candidate bounds of $\tau$.



To establish that the lower bound is sharp, Lu et al\cite{lu2018treatment}~show that for any marginal distributions $P_0$ and $P_1$ there exists a joint distribution $P$ with marginals $P_0$ and $P_1$ such that $\tau=\tau_L$. An analogous argument is used for the upper bound. In general, the upper and lower bounds are attained at different $j$s, which are determined by the marginal probabilities (see Proposition \ref{prop:tau}). In the matrix representation, this amounts to constructing a joint distribution with marginals $P_0$ and $P_1$ and with structural zeros represented as colored cells in Panels (c1) and (c2) of Figure \ref{fig:bounds} (for the appropriate $j$ levels in which the bounds are attained). The matrix representation therefore provides a convenient way to visualize the joint distributions that attain the proposed bounds. 



\begin{prop} \label{prop:eta}
    (Lu et al.\cite{lu2018treatment}~Proposition 2.) The sharp lower and upper bounds of $\eta$ are
    \begin{equation}
        \eta_L = \max_{1\le j\le J}\{\Delta_j\} \qquad
        \eta_U =  \min_{1\le j \le J} \{1 + \Delta_j - P\{Y(1)=j\}\}.
    \end{equation}
\end{prop}

Lu et al.\cite{lu2018treatment}~derive the bounds for $\eta$ by noting that
\[ \eta=1-P\{Y(0)\ge Y(1)\}, \]
which equals ``$1-\tau$'' after interchanging the roles of $Y(0)$ and $Y(1)$.

Graphically, $\eta$ can be obtained from $\tau$ by subtracting the probabilities along the diagonal $Y(0)=Y(1)$. Therefore, the lower bound for $\eta$ can be derived from the candidate lower bounds for $\tau$ after subtracting the diagonal probabilities. However, since the diagonal probabilities are not identifiable from $P_0$ and $P_1$ alone, subtraction must correspond to marginal probabilities. Consequently, each candidate lower bound indexed by $j$ for $\eta$ is obtained from the corresponding candidate lower bound for $\tau$ after subtracting either $P\{Y(0)=j\}$ or $P\{Y(1)=j\}$.

To illustrate, consider the left panel of Panel (b1) of Figure~\ref{fig:bounds} that shows the candidate lower bound for $j=3$. Subtracting the diagonal probability $P\{Y(0)=Y(1)=3\}$ can be implemented either through the marginal probability $P\{Y(0)=3\}$ or through $P\{Y(1)=3\}$. In the first case, the resulting bound equals $\Delta_{j+1}=\Delta_4$; in the second case it equals $\Delta_j=\Delta_3$. This yields the candidate lower bounds stated in Proposition~\ref{prop:eta}.

Next, consider the upper bound for $\eta$. Since $\eta \le \tau$ by definition, the upper bound for $\tau$ is also an upper bound for $\eta$. However, a tighter bound can often be obtained. Consider the candidate upper bound for $\tau$ corresponding to $j=3$, shown in Panel (b2) of Figure \ref{fig:bounds} (see Web Figure 2 for all candidate upper bounds for $\tau$). To obtain an upper bound for $\eta$, only the strictly upper-triangular cells are relevant. The matrix representation suggests that a tighter bound can be obtained by subtracting either a row or a column in such a way that all upper-triangular cells (excluding the diagonal) remain green (corresponding to summed probabilities), while none of the other cells change to orange (corresponding to subtracted probabilities). In this example, this can be achieved by subtracting either the third row or the second column, which yields the candidate upper bounds stated in Proposition~\ref{prop:eta}.

In summary, the colored matrices provide a useful tool for building intuition about the estimands, the bounds, and the differences between them. They help explain the often wide interval between the lower and upper bounds observed in real data, which can limit their practical usefulness.

\section{Bounds under positive dependence} \label{sec:DTD}
\subsection{Independence lower bounds}

The bounds derived above rely only on the marginal distributions of the potential outcomes and therefore require no assumptions about their joint distribution. 
%
%
Let $p^0_j=P\{Y(0)=j\}$ and $p^1_j=P\{Y(1)=j\}$. Under independence of $Y(0)$ and $Y(1)$, the estimands $\tau$ and $\eta$ are identifiable from the marginal distributions and are given by:
\begin{equation} \label{eq:ind}
    \tau_I=\sum_{j=1}^J \sum_{k=j}^Jp^0_jp^1_k \quad {\rm{and}} \quad \eta_I=\sum_{j=1}^{J-1} \sum_{k=j+1}^Jp^0_jp^1_k \;.
\end{equation} 
Lu et al.\cite{lu2018treatment}~suggest to use $\tau_I$ and $\eta_I$ as lower bounds for $\tau$ and $\eta$ when $Y(0)$ and $Y(1)$ are positively associated. They illustrate through several examples that these bounds can be substantially tighter than the model-free bounds in Propositions \ref{prop:tau} and \ref{prop:eta}. However, the exact dependence property guaranteeing the independent bounds hold is not discussed. Surprisingly, standard positive dependence properties are not sufficient for \eqref{eq:ind} to hold, as demonstrated below.

\subsection{Dependence notions} \label{sec:dep_indices}

Let $(U,V)$ be a pair of ordinal random variables with a joint distribution $F$. Familiar measures of association for ordinal data include Kendall's tau and Spearman's correlation. These measures, commonly used in data analysis, also have probabilistic representations. Let $(U_1,V_1),(U_2,V_2),(U_3,V_3)$ be three independent copies from $F$. Kendall's tau is defined as 
$$
{\rm Cov}\{{\rm sign}(U_2-U_1),{\rm sign}(V_2-V_1)\},
$$ 
where ${\rm sign}(a)$ equals 1 if $a>0$, -1 if $a<0$, and 0 if $a=0$. 

Spearman's correlation is proportional to \citep{lehmann1966concepts}
$$
{\rm Cov}\{{\rm sign}(U_2-U_1),{\rm sign}(V_3-V_1)\},
$$ 
and can also be written as the covariance between transformations of $U$ and $V$ according to their marginal mid-distribution functions (e.g., $F^*(u)=P(U\le u)-0.5P(U=u)$). If $U_1$ and $V_1$ tend to be small or large together, the products of the signs in the covariances above tend to be positive, indicating positive association.

A stronger notion of positive dependence is positive quadrant dependence (PQD), which compares the joint probability of large values of the variables to the probability obtained under independence with the same marginal distributions. It is defined as follows\citep{lehmann1966concepts}:
\begin{definition} \label{def:PQD}
     $U$ and $V$ are PQD if for all $j$ and $k$ in the supports of $U$ and $V$,
    \[
    P(U\ge j, V\ge k) \ge P(U\ge j)P(V\ge k).
    \]
\end{definition}
PQD implies non-negative Kendall's tau and Spearman's correlation and is therefore a stronger notion of positive dependence \citep{lehmann1966concepts}.
As PQD compares the joint distribution and the product of marginals, one might hope that this dependence concept guarantees the validity of the independence bounds in \eqref{eq:ind}. However, this is not the case. As a simple example, consider the case where $P(U=V)=1$. Then $(U,V)$ is PQD while $\eta=0$ but $\eta_I$ is positive.

PQD implies that $P(V\ge k\mid U\ge j) \ge P(V\ge k)$ (when $P(U\ge j)>0$). An even stronger notion of dependence \citep{lehmann1966concepts} is positive regression dependence (PRD) defined as follows:
\begin{definition}
    $V$ is PRD on $U$ if for all $j$ and $k$ in the supports of $U$ and $V$, $$P(V\ge j\mid U=k)$$ is non decreasing in $k$.
\end{definition}

PRD appears to be a natural dependence assumption for anyone who believes that the treatment is effective. It states that the potential outcome under treatment, $Y(1)$, stochastically increases with the potential outcome under no treatment, $Y(0)$. Equivalently, individuals who would be better off under no treatment are also expected to be better off under treatment.

However, we next show an example of a PRD random pair for which the independent bounds \eqref{eq:ind} fail for both $\eta$ and $\tau$. Thus, none of the positive dependence notions above guarantees validity of these bounds.

\subsection{Numerical examples} 

Table \ref{tab:joint} presents an example of a joint distribution of the potential outcomes for an ordinal variable with three levels. It also shows the corresponding joint distribution under independence with the same marginals. 

\begin{table}[tb]
\centering
\begin{tabular}{c c ccc @{\hspace{1cm}} ccc c}
\toprule
 &  & \multicolumn{3}{c}{joint} & \multicolumn{3}{c}{independence} & \\
\cmidrule(lr){3-5} \cmidrule(lr){6-8}
 &  &  & {$Y(0)$} &  &  & {$Y(0)$} &  & \\
 &  & 1 & 2 & 3 & 1 & 2 & 3 & $P\{Y(1)=y\}$ \\
\midrule
 & 3 & 0.03 & 0.03 & 0.19 & 0.035 & 0.030 & 0.185 & 0.25 \\
$Y(1)$ & 2 & 0.08 & 0.07 & 0.50 & 0.091 & 0.078 & 0.481 & 0.65 \\
 & 1 & 0.03 & 0.02 & 0.05 & 0.014 & 0.012 & 0.074 & 0.10 \\
\midrule
 & $P\{Y(0)=y\}$ & 0.14 & 0.12 & 0.74 & 0.14 & 0.12 & 0.74 & 1 \\
\bottomrule
\end{tabular}
\caption{Joint distribution of the potential outcomes and the corresponding joint distribution under the working independence assumption.}
\label{tab:joint}
\end{table}

Simple calculations show that 
$$
\eta=0.14<0.156=\eta_I \quad {\rm and} \quad \tau=0.43<0.433=\tau_I,
$$ 
so the working independence assumption does not provide lower bounds for either estimand. 

It is easy to verify that this distribution satisfies PRD, that is,  $P\{Y(1)\ge j \mid Y(0)=k\}$  is non-decreasing in $k$. 
Thus, PRD is not sufficient for the validity of the independence bounds, although the discrepancy is modest in this example. This need not hold in general, and, as shown below, the differences between the independent bounds and the estimands can be substantial even under PRD.

Under the no-treatment-effect case $P\{Y(0)=Y(1)\}=1$, PRD holds and $\eta=0$. Suppose that $Y(0)$, and hence also $Y(1)$, is uniformly distributed over $\{1,...,J\},$ then  $\eta_I=(J-1)/(2J).$
For example, for $J=3$ this equals 1/3, and as $J\rightarrow \infty$ it approaches 1/2.

Turning to $\tau$, let $0\le \epsilon\le 0.5$ and consider the joint distribution in Table \ref{tab:tau_example}, together with the corresponding distribution under the working independence assumption. Simple calculations show that 
$$
\tau=2\epsilon \quad {\rm and} \quad \tau_I=\frac{1+12\epsilon-4\epsilon^2}{8}.
$$ 
Thus, as $\epsilon\rightarrow 0$, $\tau \rightarrow 0$ while $\tau_I\rightarrow0.125$.

\begin{table}[tb]
\centering
\renewcommand{\arraystretch}{1.5}
\begin{tabular}{c c ccc @{\hspace{1cm}} ccc c}
\toprule
 &  & \multicolumn{3}{c}{joint} & \multicolumn{3}{c}{independence} & \\
\cmidrule(lr){3-5} \cmidrule(lr){6-8}
 &  & & {$Y(0)$} & & & {$Y(0)$} & & \\
 &  & 1 & 2 & 3 & 1 & 2 & 3 & $P\{Y(1)=y\}$ \\
\midrule
 & 3 & 0 & 0 & $\epsilon$ & $\epsilon^2$ & $\frac{(1-2\epsilon)\epsilon}{4}$  & $\frac{(3-2\epsilon)\epsilon}{4}$ & $\epsilon$ \\
$Y(1)$ & 2 &0 & 0 & $\frac{1-2\epsilon}{2}$  & $\frac{(1-2\epsilon)\epsilon}{2}$ & $\frac{(1-2\epsilon)^2}{8}$ & $\frac{(1-2\epsilon)(3-2\epsilon)}{8}$ & $\frac{1-2\epsilon}{2}$ \\
 & 1 & $\epsilon$ & $\frac{1-2\epsilon}{4}$  & $\frac{1-2\epsilon}{4}$  & $\frac{\epsilon}{2}$ & $\frac{1-2\epsilon}{8}$ & $\frac{3-2\epsilon}{8}$ & $\frac{1}{2}$ \\
\midrule
 & $P\{Y(0)=y\}$ & $\epsilon$ & $\frac{1-2\epsilon}{4}$ & $\frac{3-2\epsilon}{4}$ &  $\epsilon$ & $\frac{1-2\epsilon}{4}$ & $\frac{3-2\epsilon}{4}$ & 1 \\
\bottomrule
\end{tabular}
\caption{Joint distribution of the potential outcomes illustrating a case where $\tau<\tau_I$, together with the corresponding distribution under the working independence assumption.}
\label{tab:tau_example}
\end{table}

In summary, none of the standard positivity measures guarantees either the validity of the independence-based bounds or their proximity to the true estimands. A different justification criterion is required if one wishes to use the typically tighter independence-based bounds.

\subsection{ A new dependence notion}

The colored matrix in Panel (a) of Figure \ref{fig:bounds} suggests simple notions of dependence between $Y(0)$ and $Y(1)$ that guarantee the lower bounds $\eta_I$ and $\tau_I$ in \eqref{eq:ind}.

\begin{definition}
Let $(U,V)$ be a pair of ordinal random variables both supported on $\{1,\cdots,J\}$. We say that $V$ satisfies  open-tail diagonal tail dominance (open-tail DTD) on $U$ if for all $1\le j \le J$
\[
P(V>j|U=j)\ge P(V>j).
\]
We say that $V$ satisfies a closed-tail diagonal tail dominance (closed-tail DTD) on $U$ if for all $1\le j \le J$
\[
P(V\ge j|U=j)\ge P(V\ge j).
\]
\end{definition} 

\begin{prop} \label{prop:DTD}
Let $(Y(0),Y(1))$ be the potential outcomes of an ordinal response with support $\{1,\ldots,J\}.$
    \begin{enumerate}
    \renewcommand{\labelenumi}{(\roman{enumi})}
        \item If $Y(1)$ satisfies  open-tail DTD on $Y(0)$ then $\eta_I\le \eta$.
        \item If $Y(1)$ satisfies  closed-tail DTD on $Y(0)$ then $\tau_I\le \tau$.
    \end{enumerate}
    
\end{prop}

\begin{proof}
The open-tail DTD condition can be equivalently written as
\[
P\{Y(1)>j,Y(0)=j\}\ge P\{Y(1)>j\}P\{Y(0)=j\}.
\]
Summing over $j$ gives 
\[
\eta=\sum_jP\{Y(1)>j,Y(0)=j\}\ge \sum_j P\{Y(1)>j\}P\{Y(0)=j\} = \eta_I,
\]
which proves part (i). The proof of part(ii) is analogous. 
\end{proof}

The open- and closed-tail DTD conditions are sufficient for \eqref{eq:ind} to provide lower bounds, but they are not necessary. As  demonstrated below, there exist joint distributions for which $\tau_I<\tau$ and $\eta_I<\eta$  even though the conditions fail.

These sufficient conditions do not enforce classical positive dependence.
In particular, there exist joint distributions satisfying both open- and closed-tail DTD
for which Kendall's tau and Spearman's correlation are negative. Consequently, the sufficient conditions do not imply PQD or PRD. 

Similar to PRD, the open- and closed-tail DTD are not symmetric. This should not be surprising. Symmetric positive dependence typically means that the variables tend to be large or small together, corresponding to probability mass concentrated around the main diagonal of the joint distribution table. In contrast, large values of the estimands $\eta$ and $\tau$ reflect situations in which $Y(1)$ tends to exceed $Y(0)$, which corresponds to probability mass concentrates to the left of the main diagonal.

We conclude the section with several numerical examples. 
\begin{itemize}
    
\item Recalling that both independence bounds fail for the joint distribution in Table \ref{tab:joint}, we conclude that neither open- nor closed-tail DTD holds. Indeed, 
\[
P\{Y(1)>1\mid Y(0)=1\}=0.79<0.90=P\{Y(1)>1\},
\]
and
\[
P\{Y(1)\ge 2\mid Y(0)=2\}=0.83<0.90=P\{Y(1)\ge2\}.
\]

\item Web Tables 1-3 provide examples in which both open- and closed-tail DTD hold, only closed-tail DTD holds, and only open-tail DTD holds, respectively.  They illustrate that the independence-based bounds may fail when the corresponding DTD condition is not satisfied. 




\end{itemize}

In summary, open-tail DTD guarantees the validity of the independence-based bound for $\eta$, while closed-tail DTD guarantees the validity of the independence-based bound for $\tau$. However, neither condition is necessary. The two notions may hold simultaneously, but it is possible that one  holds while the other does not. Finally, as with many assumptions in causal inference, these conditions are untestable from the data, but may be reasonable in particular applications.

\subsection{Local diagonal tail dominance}

While the DTD criteria guarantee the validity of the independence-based bounds, they are quite strong and users may hesitate to rely on them in practice. In fact, let $S(j|k)=P\{Y(1)\ge j\mid Y(0)=k\}$ and $p^0_k=P\{Y(0)=k\}$. Then
\begin{equation} \label{eq:P(j)}
P\{Y(1)\ge j\}=\sum_{k=1}^J S(j\mid k) p^0_k.
\end{equation}
The closed-tail DTD condition holds only when the weighted average (given in \eqref{eq:P(j)}) of $S(j\mid k)$ with weights $p^0_k$ is not larger than $S(j\mid j)$. 

If PRD holds, which is a reasonable assumption for an effective treatment, $S(j\mid k)$ is increasing in $k$ for all $j$. This means that there is typically  some cutoff $\ell=\ell(j)$ satisfying an inequality of the form:
\[
S(j\mid 1)<\cdots < S(j\mid \ell)< P\{Y(1)\ge j\}<S(j\mid \ell+1) <\cdots < S(j\mid J).
\]
(we intentionally use strict inequalities to make our point clearer.)
For small $j$, the cutoff $\ell$ may be larger than $j$ and $Y(1)$ will not satisfy closed-tail DTD on $Y(0)$.

However, the same reasoning also applies to large values of $j$, for which the cutoff $\ell$ is smaller than $j$, and closed-tail DTD holds locally at $j$. In that case improved bounds are possible.

\begin{prop} \label{prop:improved}
Let $D_{ot}$ and $D_{ct}$ denote the sets of integer values at which local open-tail and local closed-tail DTD hold, respectively. That is, $j\in D_{ot}$ implies $P\{Y(1)> j\}\le P\{Y(1)> j \mid Y(0)=j\}$, and  $j\in D_{ct}$ implies $P\{Y(1)\ge j\}\le P\{Y(1)\ge j \mid Y(0)=j\}$. Then an improved lower bound for $\eta$ is
\begin{equation}
        \tilde\eta_L = \max_{1\le j\le J}\Big\{\Delta_j + \sum_{\substack{k\ge j\\ k\in D_{ot}}} P\{Y(1)> k\}P\{Y(0)=k\}\Big\},
    \end{equation}
\end{prop}
and an improved lower bound for $\tau$ is 
\begin{equation}
        \tilde\tau_L = \max_{1\le j\le J}\Big\{P\{Y(0)=j\} + \Delta_j + \sum_{\substack{k>j\\ k\in D_{ct}}} P\{Y(1)\ge k\}P\{Y(0)=k\}\Big\}.
    \end{equation}
The reasoning behind these bounds is easily seen by inspecting the colored matrices; see the right panel of Figure \ref{fig:delta}, panels (a) and (b1) of Figure \ref{fig:bounds}, and Web Figure 1. A formal proof is given in Section C of the Web Appendix. 

Thus, while using the independence-based bounds by relying on the DTD criteria may be problematic in many applications, the much weaker and more plausible assumption of local DTD  can lead to tighter bounds than those obtained without any assumptions on the dependence structure. 

\section{Real data example} \label{sec:real_data}

We apply the different bounds to data from a clinical trial studying the benefit of a new treatment for acute ischemic stroke \citep{berkhemer2015randomized}. The trial included 500 patients with acute ischemic stroke caused by a proximal intracranial arterial occlusion,  randomized to either intraarterial treatment plus usual care or usual care alone. The outcome is the modified Rankin scale (mRS) score assessed 90 days after the stroke. This is an ordinal variable with seven levels measuring functional ability, where score 0 indicates no clinical symptoms and score 6 indicates death. Because of small counts, scores 0 (no symptoms) and 1 (no clinically significant disability) are combined. Finally, to remain consistent with the discussion above in which higher levels of the ordinal scale correspond to better outcomes, we define score 6 of the mRS as level 1 of our ordinal variable and scores 0-1 of the mRS as level 6. 

The original study used multivariate ordinal logistic regression to estimate an adjusted odds-ratio,  finding a significant benefit of intraarterial treatment, with an odds ratio of 1.67. 
We reanalyze the data, calculating model-free estimands and bounds (see Section \ref{sec:discussion} for a short discussion regarding inference).  Table \ref{tab:stroke} presents the data and the resulting empirical probabilities. The last row reports $\Delta_j$, which is the average causal effect obtained when the ordinal outcome is replaced by a binary indicator of the form ${\rm I}\{Y\ge j\}$. The new treatment is beneficial for all $j$, with the additive effect varying across levels. 

\begin{table}[tb]
    \centering 
    \begin{tabular}{lcccccc}
\toprule
mRS score &         0-1 &  2 &  3 & 4 & 5 &  6 \\
Ordinal level ($j$) & 6 & 5 & 4 & 3 & 2 & 1 \\
\midrule
\# trt   &  27 & 49 & 43 & 52 &  13 & 49 \\
\# ctrl  &   16 & 35 & 44 & 81 & 32 & 59 \\
$P(Y=j\mid {\rm trt})$ &  0.116 & 0.210 & 0.185 & 0.223 & 0.056 & 0.210 \\
$P(Y=j\mid {\rm ctrl})$ & 0.060 & 0.131 & 0.165 & 0.303 & 0.120 & 0.221 \\
$P(Y\ge j\mid {\rm trt})$ &  0.116 & 0.326 & 0.511 & 0.734 & 0.790 & 1  \\
$P(Y\ge j\mid {\rm ctrl})$ &  0.060 & 0.191 & 0.356 & 0.659 & 0.779 & 1 \\
$\Delta_j$  &  0.056 & 0.135 & 0.155 & 0.075 & 0.011 & 0 \\
\bottomrule
\end{tabular}
\caption{Data and statistics from the stroke clinical trail. First two rows show the mRS score and corresponding level for the analysis. Rows 3 and 4 are the raw counts for the treatment (trt) and control (ctrl) groups.  Rows 5-9 provide statistics for calculating the bounds.}
\label{tab:stroke}
\end{table}

Using the model-free sharp bounds in Propositions \ref{prop:tau} and \ref{prop:eta}, we find that the proportion of individuals for whom  the treatment is beneficial, $\eta$, lies between 0.155 and 0.790, and the proportion for whom the treatment is non-harmful, $\tau$, lies between 0.378 and 1. These intervals are quite wide, making the bounds only weakly informative. Under a working independence assumption, the estimands in \eqref{eq:ind}  are 0.486 and 0.672. Using these as lower bounds provides a much sharper conclusion. They imply that the proportion of individuals who benefit from the new treatment is at least about 50\% and that the proportion   who are not harmed is substantially greater than 50\%. According to Proposition \ref{prop:DTD}, these lower bounds are valid under the untestable assumption of open- and closed-tail DTD. In the present example, this assumption means that knowing that a person's outcome under usual care is $j$ would increase the chance of having at least as good an outcome (closed-tail) or a strictly better outcome (open-tail) when having the option of intraarterial treatment, rather than usual care alone.

However, as discussed above, this assumption is problematic for small $j$, say for $j=3$, which is the modal score under both treatment and control. It seems safer to assume local DTD for large values, say 4 and above. Doing so and using the bounds in Proposition \ref{prop:improved}, we obtain the bounds $\tilde \eta=0.224$ and $\tilde \tau=0.512$. These improve on the sharp bounds that make no assumptions about the dependence structure.

\section{Discussion} \label{sec:discussion}

Causal inference for ordinal outcomes is challenging, and no fully satisfactory general approach has yet gained consensus. Regression models, such as the proportional-odds model, yield estimates that are not easily interpreted in the potential-outcomes framework because of the noncollapsibility of the odds ratio \citep{whitcomb2021defining}. Other approaches modify the nature of the outcome by binarizing it or by treating it as continuous through the assignment of numerical scores; such approaches either lose information or impose arbitrary assumptions. In contrast, the approach considered here defines estimands based on the joint distribution of the potential outcomes. Although these estimands are generally not identifiable, they admit identifiable bounds. Without additional assumptions, these bounds may be loose, but they can be substantially tightened under dependence assumptions. We show that standard “weak” notions of dependence are insufficient to guarantee the validity of independence-based bounds\cite{lu2018treatment}, and we introduce a new stronger dependence criterion that is sufficient.

Colored matrices provide a simple and effective tool for visualizing causal estimands for ordinal outcomes and for formulating dependence notions that guarantee the validity of improved bounds.
The technique and the associated bounds extend naturally to discrete numerical variables and, with suitable modifications, to continuous variables by replacing the matrices with corresponding regions in the $(X,Y)$ plane (note that the ordering of values in the matrices is consistent with the standard orientation of the coordinate axes). Fan and Park\cite{fan2010sharp} derive bounds for the continuous case, in which $\eta=\tau$. Their results are equivalent to those presented in Propositions \ref{prop:tau} and \ref{prop:eta} after omitting the marginal probabilities $P\{Y(0)=j\}$ and $P\{Y(1)=j\}$, which vanish in the continuous setting.

The notions of diagonal tail dominance and local diagnoal tail dominance can be extended to the continuous case with appropriate modifications. In this setting, there is no distinction between open- and closed-tail versions. We say that $Y(1)$ satisfies diagonal tail dominance (DTD) on $Y(0)$ if there exists a version of the conditional distribution of $Y(1)$ given $Y(0)$ such that 
\[
P\{Y(1)>y\mid Y(0)=y\} \ge P\{Y(1)>y\}
\]
for almost all $y$. Under this definition, and the corresponding definition of a local DTD, Propositions \ref{prop:DTD} and \ref{prop:improved} extend to the continuous case.

In the presence of confounders, the bounds can be improved\cite{lu2018treatment} by calculating covariate-specific bounds and taking expectations with respect to the distribution of the confounders, denoted by $Z$. Specifically,
\[
\tilde\eta_L=\int \eta_L^z \, dF_Z(z),\quad \tilde\eta_U=\int \eta_U^z \, dF_Z(z), \quad 
\tilde\tau_L=\int \tau_L^z \, dF_Z(z), \quad \tilde\tau_U=\int \tau_U^z \, dF_Z(z),
\]
where $F_Z$ denotes the distribution of the confounders, and the superscript $z$ indicates that the bounds are computed with respect to the conditional distribution of $(Y(0),Y(1))\mid Z=z$. The lower bounds remain valid if (local) open- and closed-tail DTD hold conditionally on the confounders; that is, if these conditions hold within each subpopulation defined by $Z=z$.

In practice, as illustrated in Section \ref{sec:real_data}, the marginal distributions are not known and must be inferred from the data. The bounds are therefore estimates and are subject to sampling uncertainty. Since they are functions of the marginal empirical distributions, standard methods can be used to estimate their sampling variance and to construct confidence intervals; see Lu et al.\cite{lu2018treatment}~for discussion and examples.

Finally, we emphasize that both PRD and DTD inherently assume that the treatment is more favorable than the control, and therefore it is not suitable for comparing a new treatment to a standard one. In such settings, one should first test for equality between the treatment arms. Only if this hypothesis is rejected the local or global DTD assumption should be considered as a justification for using improved bounds on the causal parameters.

\section*{Acknowledgments}
We thank Daniel Nevo and Yosef Rinott for helpful comments and suggestions. 


\bibliographystyle{unsrtnat}
\bibliography{references}


\section*{Supporting information}
Additional supporting information can be found  in the Supporting Information section.
R functions that reproduce all figures and tables presented in this paper is available at \url{https://github.com/MichaMandel/Ordinal-bounds}. 

\appendix
\counterwithin{table}{section}
\counterwithin{figure}{section}

\section{Supporting Information -- Additional figures and tables}

\begin{figure}[p]
\centering
\includegraphics[width=\textwidth,height=0.85\textheight,keepaspectratio]{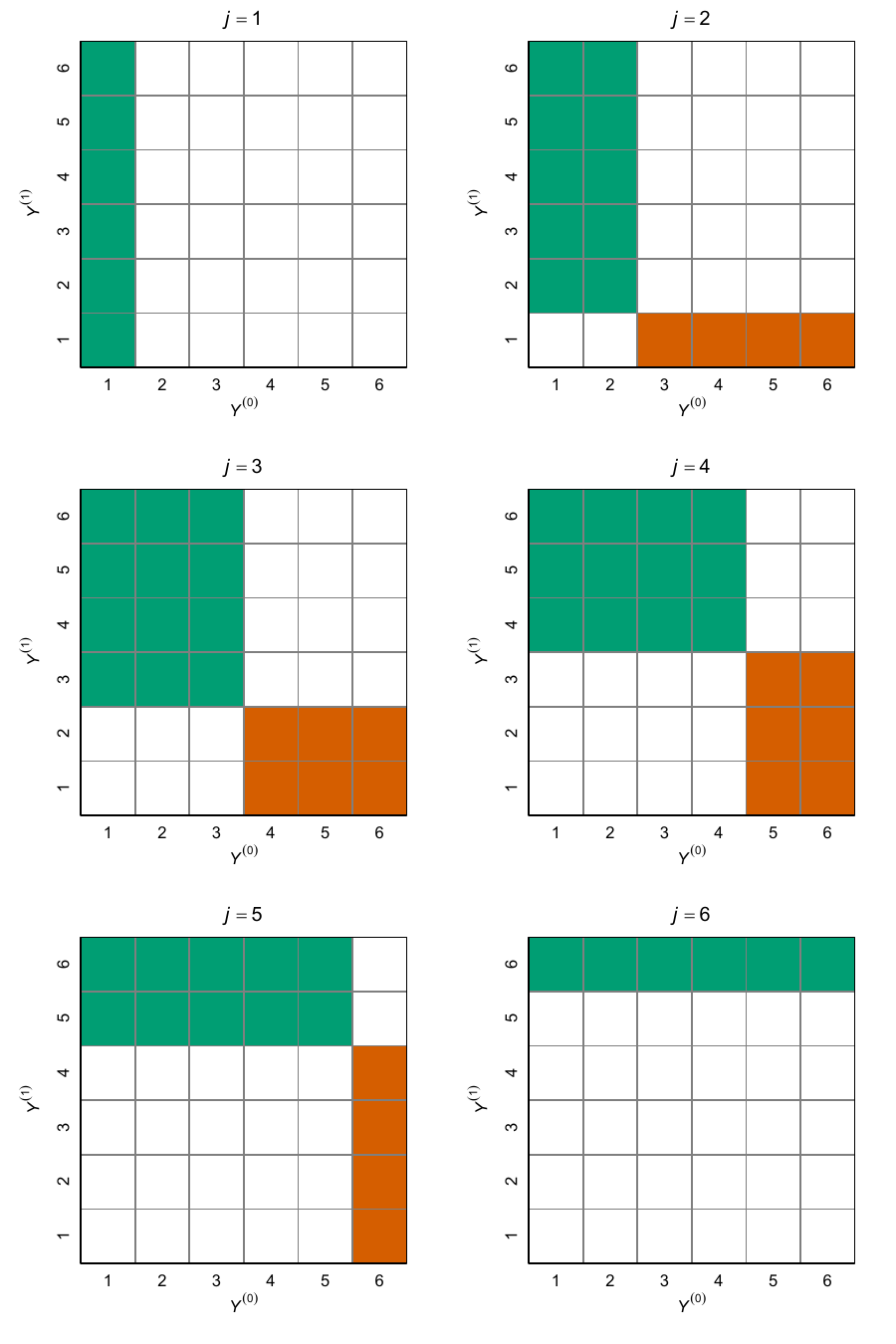}
\caption{Lower bounds for $\tau$. The matrix represents the joint probability table of $(Y(0),Y(1))$. The bounds sum the green entries and subtract the red ones. }
\label{fig:lower_tau_all}
\end{figure}

\begin{figure}[p]
\centering
\includegraphics[width=\textwidth,height=0.85\textheight,keepaspectratio]{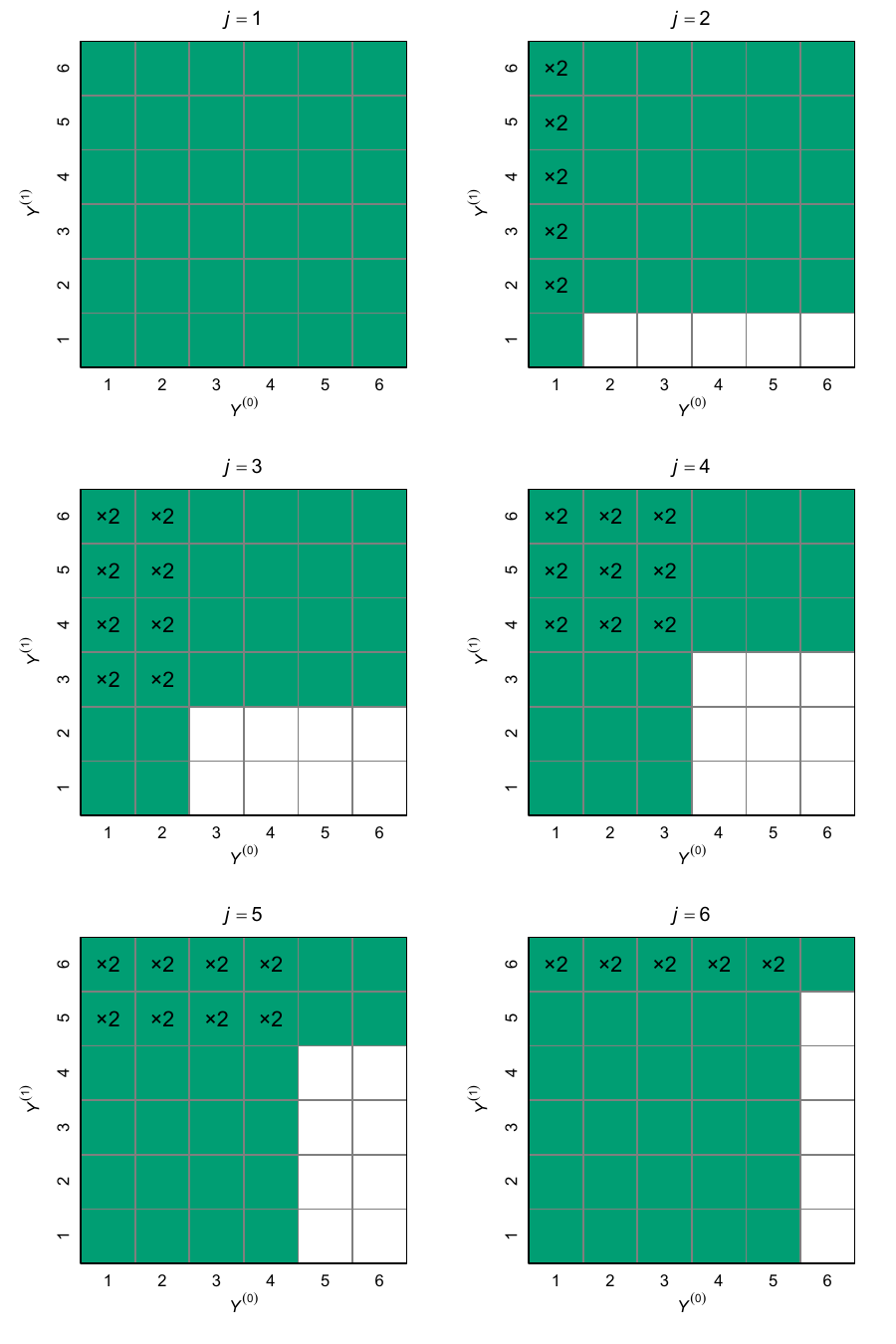}
\caption{Upper bounds for $\tau$. The matrix represents the joint probability table of $(Y(0),Y(1))$. The bounds sum the green entries; entries with $\times 2$ are summed twice.}
\label{fig:uper_tau_all}
\end{figure}

\clearpage

\begin{table}[ht]
\centering
\begin{tabular}{c c ccc @{\hspace{1cm}} ccc c}
\toprule
 &  & \multicolumn{3}{c}{joint} & \multicolumn{3}{c}{independence} & \\
\cmidrule(lr){3-5} \cmidrule(lr){6-8}
 &  &  & {$Y(0)$} &  &  & {$Y(0)$} &  & \\
 &  & 1 & 2 & 3 & 1 & 2 & 3 & $P\{Y(1)=y\}$ \\
\midrule
 & 3 & 0.13 & 0.13 & 0.09 & 0.14 & 0.119 & 0.091 & 0.35 \\
$Y(1)$ & 2 & 0.11 & 0.1 & 0.05 & 0.104 & 0.0884 & 0.0676 & 0.26 \\
 & 1 & 0.16 & 0.11 & 0.12 & 0.156 & 0.1326 & 0.1014 & 0.39 \\
\midrule
 & $P\{Y(0)=y\}$ & 0.4 & 0.34 & 0.26 & 0.4 & 0.34 & 0.26 & 1 \\
\bottomrule
\end{tabular}
\caption{Joint distribution of the potential outcomes illustrating a case where closed and open-DTD hold, together with the corresponding distribution under the working independence assumption.}
\label{tab:DTDhold}
\end{table}

\begin{table}[ht]
\centering
\begin{tabular}{c c ccc @{\hspace{1cm}} ccc c}
\toprule
 &  & \multicolumn{3}{c}{joint} & \multicolumn{3}{c}{independence} & \\
\cmidrule(lr){3-5} \cmidrule(lr){6-8}
 &  &  & {$Y(0)$} &  &  & {$Y(0)$} &  & \\
 &  & 1 & 2 & 3 & 1 & 2 & 3 & $P\{Y(1)=y\}$ \\
\midrule
 & 3 & 0.07 & 0.18 & 0.1 & 0.1225 & 0.161 & 0.0665 & 0.35 \\
$Y(1)$ & 2 & 0.08 & 0.22 & 0.05 & 0.1225 & 0.161 & 0.0665 & 0.35 \\
 & 1 & 0.2 & 0.06 & 0.04 & 0.105 & 0.138 & 0.057 & 0.3 \\
\midrule
 & $P\{Y(0)=y\}$ & 0.35 & 0.46 & 0.19 & 0.35 & 0.46 & 0.19 & 1 \\
\bottomrule
\end{tabular}
\caption{Joint distribution of the potential outcomes illustrating a case where closed but not open-DTD hold and $\eta<\eta_I$, together with the corresponding distribution under the working independence assumption.}
\label{tab:eta<eta_I}
\end{table}

\begin{table}[ht]
\centering
\begin{tabular}{c c ccc @{\hspace{1cm}} ccc c}
\toprule
 &  & \multicolumn{3}{c}{joint} & \multicolumn{3}{c}{independence} & \\
\cmidrule(lr){3-5} \cmidrule(lr){6-8}
 &  &  & {$Y(0)$} &  &  & {$Y(0)$} &  & \\
 &  & 1 & 2 & 3 & 1 & 2 & 3 & $P\{Y(1)=y\}$ \\
\midrule
 & 3 & 0.2 & 0.1 & 0.05 & 0.1925 & 0.0875 & 0.07 & 0.35 \\
$Y(1)$ & 2 & 0.3 & 0.05 & 0.05 & 0.22 & 0.1 & 0.08 & 0.4 \\
 & 1 & 0.05 & 0.1 & 0.1 & 0.1375 & 0.0625 & 0.05 & 0.25 \\
\midrule
 & $P\{Y(0)=y\}$ & 0.55 & 0.25 & 0.2 & 0.55 & 0.25 & 0.2 & 1 \\
\bottomrule
\end{tabular}
\caption{Joint distribution of the potential outcomes illustrating a case where open but not closed-DTD hold and $\eta<\eta_I$, together with the corresponding distribution under the working independence assumption.}
\label{tab:tau<tau_I}
\end{table}

\end{document}